\documentclass{article}
\usepackage{graphicx} 
\usepackage[russian, english]{babel}
\usepackage[T1, T2A]{fontenc}
\usepackage[utf8]{inputenx}

\usepackage{tikz-cd}

\usepackage{amssymb}
\usepackage{stmaryrd}
\usepackage{amsmath}
\usepackage{amsthm}

\newtheorem{theorem}{Theorem}
\newtheorem{lemma}{Lemma}
\newtheorem{corollary}{Corollary}
\newtheorem{proposition}{Proposition}
\theoremstyle{definition}
\newtheorem{definition}{Definition}
\newtheorem{example}{Example}
\newtheorem{remark}{Remark}

\DeclareMathOperator*{\argmax}{argmax}

\title{On possible values of the group complexity function of infinite words}
\author{Maksim Launer, Svetlana Puzynina, Ekaterina Voloshinova \\
Saint Petersburg State University, Russia\\
mlauner\texttt{\_}official@bk.ru, s.puzynina@gmail.com, evoloshinova@gmail.com}
\date{}

\def\cb#1{{\color{blue} #1}}

\begin{document}

\maketitle

\begin{abstract} 
A classical notion of a factor complexity of an infinite word is defined as a function $p(n)$ counting, for each $n$, the number of distinct factors (or blocks of consecutive letters) of the word of length $n$. The notion has various generalizations and variants. For example, the abelian complexity $p_{ab}(n)$ counts the number of distinct factors of each length $n$ up to abelian equivalence, i.e., only the numbers of occurrences of letters are taken into account, and not their order. The notion of a group complexity generalizes both notions of a factor and an abelian complexities. Namely, given a sequence $\omega=(G_n)_{n=1}^{\infty}$ of subgroups of the symmetric group $S_n$, the group complexity $p_{\omega}(n)$ of a word counts the number of classes of factors of each length $n$ of the word, where words obtained from one another by permutations from $G_n$ are put in the same class. Taking $G_n=S_n$, we obtain the abelian complexity, and taking $G_n=Id$, we recover the factor complexity. Clearly, the group complexity value is between the abelian and the factor complexities. In this paper, we are interested in the following property of words. We say that an infinite word has universal group complexity if for each length $n$ and for each $k$ satisfying $p_s^{ab}(n) \leqslant k \leqslant p_s(n)$, there exists a group $G \in S_n$ such that $p_s^G(n) = k$. In other words, all ``intermediate'' values of complexity can be obtained. We show that Sturmian words satisfy the universal group complexity property, while they are not the only ones. We also study the universal group complexity property for aperiodic ternary words of minimal complexity and for eventually periodic words.
\end{abstract}


\section{Introduction}
	
For each infinite word $w$, its factor complexity function counts, for each $n$, the number of distinct factors of $w$ of length $n$. This notion  was introduced in 1938 in a seminal paper by Morse and Hedlund \cite{mih} on symbolic dynamics. Among other results, Morse and Hedlund gave a relation between factor complexity and periodicity in infinite words; namely, they proved that each aperiodic infinite word $w$ has factor complexity at least $n + 1$ for each length $n$. They further showed that an infinite word $w$ has complexity $n+1$ for each length $n$ if and only if $w$ is binary, aperiodic and balanced, i.e., $w$ is a Sturmian word (see also \cite{abelian,MoHe40}). Thus Sturmian words are those aperiodic words of the lowest factor complexity. They arise naturally in many different areas of mathematics including combinatorics, algebra, number theory, ergodic theory and dynamical systems. Sturmian words also have applications in theoretical physics as 1-dimensional models of quasi-crystals, and in theoretical computer science where they are used in computer graphics as digital approximation of straight lines. 
For more on Sturmian words, we refer to Chapter 2 in \cite{Lothaire_2002}.

Problems in the study of factor complexity of infinite words include characterizing complexities of important families of words, such as morphic \cite{pansiot} and Toeplitz \cite{CK97} words, and the study of words of linear complexity \cite{CFPZ19,Leroy2012}.
An important longstanding open problem in combinatorics on words is an inverse problem of characterizing possible complexity functions of infinite words; see, e.g., \cite{DBLP:conf/dlt/Cassaigne95} and a recent characterization in an asymptotic form \cite{zelmanov}. For more on factor complexity we refer to a book chapter \cite{CasNic}. 

There exist multiple generalizations and extensions of the notion of a complexity function including the abelian complexity \cite{abelian09}. Two finite words are said to be abelian equivalent if they are permutations of each other. In other words, in abelian combinatorics on words we consider commutative images of words, so that the order of letter is not taken into account. The abelian complexity of a word is defined as a function counting the number of abelian classes of factors of a word for each length. Sturmian words also represent the family of aperiodic words of the smallest abelian complexity. For a survey on abelian properties of words we refer to \cite{DBLP:journals/csr/FiciP23}.

The notion of a group complexity introduced in \cite{puzc} generalizes the notions of a factor and an abelian complexities. Consider the group action of a permutation group $G \leqslant S_n$ acting on the set of words of length $n$ on a finite alphabet $A$ in a natural way. In order to define a group complexity, we consider infinite sequences of permutation groups $\omega = (G_n)_{n\geq 1}$ with each $G_n \subseteq S_n$. Associated with every such sequence, and with every infinite word $x$, the group complexity function $p_x^{\omega}$ counts for each length $n$ the number of equivalence classes of factors of $x$ of length $n$ under the action of $G_n$ on words of length $n$. As shown in \cite{puzc}, group complexity also admits an analog of Morse and Hedlund theorem. 
    
Factor and abelian complexities are particular cases of group complexity: Taking the sequence of trivial groups, i.e, $G_n=Id_n$, we obtain factor complexity, and taking the sequence of symmetric groups, i.e, $G_n=S_n$, we recover  abelian complexity. The definition of a group complexity implies that for each group $G_n$, we have $p^{ab}(n) \leqslant p^{G_n}(n) \leqslant p(n)$. In this paper, we are interested in when all the intermediate values are achieved. Namely, we say that a word satisfies the universal group complexity property if for each length and for each integer between the abelian and the factor complexities there exists a subgroup of the symmetric group giving this group complexity value. We prove that Sturmian words satisfy the universal group complexity property, while many other words including the Thue-Morse word do not have universal group complexity. However, this does not give a characterization of Sturmian words, as some other words, e.g., certain images of Sturmian words, also have universal group comlpexity. We study group complexity values for ternary words of minimal complexity, and provide a characterization of universal group complexity property for eventually periodic words.

	\section{Preliminaries}
    
        \subsection{Words and their complexity functions}
	An \emph{alphabet} $A$ is a finite set of symbols; its elements are called \emph{letters}. A \emph{word} over the
alphabet $A$ is a finite or infinite sequence of  letters from $A$. For a finite word $v$, its \emph{length} $|v|$ is the number of letters in it. The
	\emph{empty word} $\varepsilon$ is the word containing no letters; by convention, we set  $|\varepsilon| = 0$. The number of occurrences of a letter $a$ in a finite word $v$ is denoted by $|v|_a$.
	
	A word $u$ is a \emph{factor} of a word $w$ if $w = xuv$ for some words $x$ and $v$. If in addition $x = \varepsilon$ (resp., $v = \varepsilon$), then $u$ is called a \emph{prefix} (resp., \emph{suffix}). The set of factors of a word $w$ of length $n$ is denoted by $F_w(n)$. A factor $u$ of a word $w$ is called \emph{right} (resp., \emph{left}) special if $ua$ and $ub$ (resp., $au$ and $bu$) are factors of $w$ for some distinct letters $a$ and $b$. A factor $u$ of an infinite word $w$ is called \emph{bispecial} if it is left special and right special. 
  
	Given an order $\prec$ on the alphabet $A$, we can extend it to the \emph{lexicographic order} on words as follows. For words  $u$ and $v$  on $A$ we have $u < v$ if either $u$ is a prefix of $v$ or $|u| = |v|$ and there exists a prefix $x$ of $u$ and $v$ such that  $u = xay, v = xbz$, where $y$ and $z$ are words and $a$ and $b$ are letters such that $a \prec b$.
 	
    
    Given an infinite word $u$, its \emph{(factor) complexity} is a function $p_u: \mathbb{N} \to \mathbb{N}$, counting, for each integer $n$, the number of distinct factors of $u$ of length $n$.     Clearly, the complexity function of a word satisfies some straightforward properties. For example, it must be non-decreasing, and it is bounded from above by $|A|^n$. However, a complete characterization of possible values of complexity is an open question.

 Two words $v$ and $u$ are called \emph{abelian equivalent}, denoted by $v \sim_{ab} u$, if for each  $a \in A$ we have $|v|_a = |u|_a$. In other words, $v$ and $u$ can be obtained from one another by a permutation of letters. It is straightforward that abelian equivalence is indeed an equivalence relation
on the set of finite words. The \emph{abelian complexity} of a word $w$ is a function  $p_w^{ab}: \mathbb{N} \to \mathbb{N}$, counting for each $n$ the number of abelian classes of factors of $w$ of length $n$, i.e., $p_w^{ab}(n) = |F_w(n)/{\sim_{ab}}|$. 
	
Consider a subgroup $G \leqslant S_n$ of the symmetric group $S_n$. It acts on words of length $n$ by permutation of letters in a natural way: for  $g \in G$, we have
	$$g(a_1\cdots a_n) = a_{g^{-1}(1)}\cdots a_{g^{-1}(n)}.$$
	Each subgroup $G \leqslant S_n$ defines an equivalence relation $\sim_G$ on the set of words of length $n$, given by $u \sim_G v$ if there exists $ g \in G$ such that $g(v) = u$. 
In other words, $u\sim_G v$ if and only if $u$ and $v$ are in the same $G$-orbit relative to the action of $G$ on $A^n.$ We can now define the group complexity of an infinite word:

\begin{definition} Let $w$ be an infinite word and $\omega=(G_n)_{n\geq 1}$ a sequence of subgroups $G_n$ of $S_n$. The associated \emph{group complexity function} $p_w^{\omega}:\mathbb{N} \rightarrow \mathbb{N}$ counts, for each length $n$, the number of $\sim_{G_n}$ equivalence classes of factors of length $n$ of an infinite word $w$, i.e.  $p_w^\omega(n) = |F_w(n)/{\sim_{G_n}}|$.
\end{definition}

The group complexity is a generalization of both factor complexity and abelian complexity: Indeed, taking $\omega=(Id_n)_{n\geq 1}$, we obtain factor complexity, and taking $\omega=(S_n)_{n\geq 1}$, we recover  abelian complexity. 

It is straightforward that for each $G \leqslant S_n$ and for each word $w$ the value of its group complexity relatively to the group $G_n$ is between its abelian and factor complexities: 
	$$p_w^{ab}(n) \leqslant p_w^{G_n} \leqslant p_w(n).$$
In this paper, we are interested in the following question: for which words all possible values between abelian and factor complexity can be achieved? In other words, we introduce and study the following property of infinite words:
\begin{definition} An infinite word $w$ has \emph{universal group complexity} property if, for all integers $n$ and $m$, $p_w^{ab}(n) \leqslant m \leqslant p_w(n)$, there exists a subgroup $G\leqslant S_n$ such that $p_w^G = m$. If the property holds for a given length $n$, we say that $w$ has \emph{universal group complexity} property \emph{for length $n$}. 
\end{definition}

   In the paper, we will use the following notation for certain subgroups of $S_n$:
\begin{itemize}
\item $S_{[i,j]}$: the subgroup permuting all elements from $i$ to $j$ and fixing elements $1, \dots, i-1$ and $j+1,\ldots, n$. I.e., the subgroup is isomorphic to $S_{j-i+1}$.
\item $C_n$: the cyclic group of order $n$ generated by the cycle $(12\dots n)$.
\end{itemize}

A mapping $\phi: \Sigma^* \rightarrow \Sigma^*$ is called a \emph{morphism} if it preserves concatenation; that is, for each pair of words $u$, $v$ we have $\phi(uv) = \phi(u)\phi(v)$. Since a morphism is completely defined by its images on letters, the definition of a morphism can be naturally extended to  infinite words. 
A morphism $\phi$ is \emph{nonerasing} if $\phi(a) \neq \varepsilon$ for each letter $a \in \Sigma$.
 A morphism $\phi$ is called \emph{prolongable on} $a\in \Sigma$ if $\phi(a) = au$ for some $u \in \Sigma^+$. If a morphism is nonerasing and prolongable on $a$, it defines a unique infinite word $x=\phi^{\omega}(a)$, where the limit is taken in the prefix sense; $x$ is said to be \emph{generated by} $\phi$.

   \subsection{Aperiodic words of minimal complexity}

	 An infinite word is called \emph{eventually periodic} if it is of the form $uvvvvv\dots$, where $u, v$ are finite words. If $u = \varepsilon$, then the word is called \emph{(purely) periodic}. A word is called \textit{aperiodic} if it is not eventually periodic.
     
    A celebrated theorem of Morse and Hedlund gives a link between periodicity
and factor complexity:

\begin{theorem}[\cite{mih}] Let $w$ be an infinite word. If there exists $n$ such that $p_w(n) \leqslant n$, then $w$ is eventually periodic. \end{theorem}

Moreover, the upper bound from the theorem cannot be improved: there exist aperiodic words of complexity $n+1$ for each $n$, and such words are called \emph{Sturmian} \cite{MoHe40}. This means that Sturmian words can be regarded as  the simplest aperiodic words. Since for a Sturmian word $s$ we have $p_s(1)=2$, these words are binary. For general alphabets, minimal complexity of aperiodic words on an alphabet $A$ is  $p(n) = n + |A| - 1$ for all $n$. We discuss in detail ternary words of minimal complexity in Subsection~\ref{subsec:ternary}. Sturmian words and more generally aperiodic words of minimal complexity have exactly one right (and one left) special factor for each length, and this factor is extended by all letters of the alphabet.

	
	

	The following is known for every Sturmian word (\cite{abelian}): $p^{ab}(n) = 2$ for all $n$. For each length $n$ two abelian classes differ in number of 1's, one class has words with more 1's (it is called \emph{rich}) than the words in other class (it is called \emph{poor}).

    \subsection{Lexicographic arrays of Sturmian words} \label{subsec:lex_array}

In this section we introduce \emph{lexicographic arrays} of Sturmian words, which we will use throughout the paper.

Let $w$ be a Sturmian word and $n$ be an integer. Consider all factors of length $n$ of $w$ in lexicographic order. As shown in \cite{DBLP:journals/tcs/JenkinsonZ04} and \cite{DBLP:journals/tcs/PerrinR12}, 
consecutive factors $v$ and $v'$ are of the form either $v = x 01 y$, $v' = x 10 y$ for some finite (possibly empty) words  $x, y$, or $v = x 0$, $v' = x1 $. The latter case corresponds to the situation when $x$ is right special. Since Sturmian words have exactly one right special factor for each length, the second case occurs exactly once. In particular, it follows that the two abelian classes of factors of a Sturmian word of length $n$ are grouped together: first we have factors poor in 1, and then factors rich in 1. Moreover, the switch from the class of poor words to the class of rich words occurs exactly in the place between  $v = x 0$ and $v' = x1 $.

\begin{example} Here we give an example of the lexicographic array for the Fibonacci word for $n = 10$. The Fibonacci word $f$ is the most famous example of a Sturmian word. It can be defined for example by iterating a morphism $\varphi: 0\to 01, 1\to 0$: $f=\varphi^{\infty}(0)=01001010010\cdots$. For each factor, we mark by green symbols that change compared to the next factor, and we mark the corresponding symbols in the next factor by red. The horizontal line splits the abelian classes.

\[\begin{matrix}
0&0&1&0&{\color{green}0}&{\color{green}1}&0&1&0&0 \\
0&0&1&0&{\color{red}1}&{\color{red}0}&0&1&0&{\color{green}0} \\
\hline
0&{\color{green}0}&{\color{green}1}&0&1&0&0&1&0&{\color{red}1}\\
0&{\color{red}1}&{\color{red}0}&0&1&0&{\color{green}0}&{\color{green}1}&0&1\\
0&1&0&{\color{green}0}&{\color{green}1}&0&{\color{red}1}&{\color{red}0}&0&1\\
0&1&0&{\color{red}1}&{\color{red}0}&0&1&0&{\color{green}0}&{\color{green}1}\\
{\color{green}0}&{\color{green}1}&0&1&0&0&1&0&{\color{red}1}&{\color{red}0}\\
{\color{red}1}&{\color{red}0}&0&1&0&{\color{green}0}&{\color{green}1}&0&1&0\\
1&0&{\color{green}0}&{\color{green}1}&0&{\color{red}1}&{\color{red}0}&0&1&0\\
1&0&{\color{red}1}&{\color{red}0}&0&1&0&{\color{green}0}&{\color{green}1}&0\\
1&0&1&0&0&1&0&{\color{red}1}&{\color{red}0}&0\\
\end{matrix} \]
\end{example}

\subsection{Simple continued fractions and Sturmian words} \label{subsec:cont_frac}

We recall that every positive real number $\alpha$ has a representation as a simple continued fraction, finite or infinite, as follows: $$\alpha = a_0 + \dfrac{1}{a_1 + \dfrac{1}{a_2 + \dfrac{1}{\ddots}}} \textrm{\quad or\quad} \alpha = a_0 + \dfrac{1}{a_1 + \dfrac{1}{a_2 + \dfrac{1}{\ddots + \dfrac{1}{a_n}}}}.$$
Such representation is briefly written as $\alpha=[a_0; a_1, a_2, \ldots]$. It is well known that $\alpha$ has a finite simple continued fraction if and only if $\alpha$ is rational. \emph{Convergents} of $\alpha$ are defined as fractions $\frac{p_m}{q_m}$ as follows: $\frac{p_0}{q_0} = a_0, \frac{p_1}{q_1} = a_0 + \frac{1}{a_1}, \frac{p_2}{q_2} = a_0 + \frac{1}{a_1 + \frac{1}{a_2}}, \dots$   

A \emph{slope} of a Sturmian word can be defined, e.g., as the frequency of 1's in it.
Given a Sturmian word $u$ of slope $\alpha < 1$,  one can consider a \emph{standard sequence} of the word $u$, i.e., a sequence $(s_n)_{n\geqslant -1}$ defined as follows: $$s_n = \begin{cases}
    1, & n = -1,\\
    0, & n = 0,\\
    s_{n-1}^{d_n}s_{n-2}, & \mathrm{otherwise}
\end{cases}$$ where $\alpha = [0, d_1, d_2, \dots]$ is the continued fraction expansion. The elements of the sequence $s_n$ are factors of $u$ and they are called \emph{standard}. For more on standard and bispecial factors of Sturmian words we refer to \cite{CASSAIGNE201736} and Paragraph 2.2 in \cite{Lothaire_2002}.

The covergents of $\alpha$ satisfy the following recurrent formulas: 
\begin{equation*}\begin{split} & p_m = d_m p_{m-1} + p_{m - 2}, \\ & q_m = d_m q_{m-1} + q_{m - 2};\end{split}\end{equation*} for $m \geqslant 1$, where $p_{-1} = 1$ and $q_{-1} = 0$. This implies that for a Sturmian word of slope $\alpha$ we have $|s_n|_1 = p_n, |s_n|_0 = q_n$, and hence, $|s_n| = p_n + q_n$. For convergents, the following equality holds: \begin{equation} \label{eq:convergents} p_m q_{m-1} - q_m p_{m-1} = (-1)^{m - 1}.\end{equation} The above equalities they are well-known and can be easily proved by induction; they are also implicitly contained in Paragraph 2.2. in \cite{Lothaire_2002}.  

\section{On group complexity of Sturmian words}

The main result of this section is the following

\begin{theorem}\label{th:sturmian}
Sturmian words have universal group complexity.
\end{theorem}

\begin{proof} 
Let $w$ be a Sturmian word and $n$ be an integer. Consider the lexicographic array of $w$ for length $n$ (see Subsection \ref{subsec:lex_array}). 

For an integer $m$, $1\le m < n - 1$, consider the subgroup $S_{[m+1,n]}$ (recall that this denotes the group that does not change the first $m$ letters and acts as $S_{n-m}$ on the last $(n-m)$ letters). 

Among prefixes of length $m$ we have exactly   $m + 1$ distinct words, since these prefixes are factors of $w$ of length $m$. So, $p^{S_{[m+1,n]}}(n) \ge m + 1$. We now show that in fact we have $p^{S_{[m+1,n]}}(n) = m + 2$. 

First we show this for $m=0$. For $m = 0$, we have the action of the group $S_n$. Group complexity relatively to $S_n$ is abelian complexity, and for Sturmian words it is equal to $2$, which is $m + 2$ for $m = 0$. 

Now we compare classes of group equivalence of factors relatively to subgroups $S_{[m+1,n]}$ and $S_{[m+2,n]}$. Clearly, the partition corresponding to $S_{[m+2,n]}$ is a refinement of that of $S_{[m+1,n]}$, since $S_{[m+2,n]}$ is a subgroup of $S_{[m+1,n]}$.

We claim that each group equivalence class relatively to $S_{[m+1,n]}$ contains several consecutive words from the lexicographic array.

We prove this by induction on $m$. Consider an equivalence class $L$ relatively to $S_{[m+1,n]}$  which is split relatively to $S_{[m+2,n]}$. Then there are two consecutive factors $v$ and $v'$ from $L$ which are in different equivalence classes relatively to $S_{[m+2,n]}$. Consider  prefixes of $v$ and $v'$ of length $m$ and $m+1$. Note that $v$ and $v'$ are abelian equivalent, since they are in the same group equivalence class relatively to $S_{[m+1,n]}$.

Since $v$ and $v'$ are in the same class of equivalence relatively to  $S_{[m+1,n]}$  and this group fixes the first $m$ indices, the prefixes of $v$ and $v'$ of length $m$ coincide, i.e., $v = uw$ and $v' = uw'$, where $|u|=m$. Then $w\sim_{ab} w'$, since $v$ and $v'$ are abelian equivalent. If $v_{m+1}=v'_{m+1}$, then $v$ and $v'$ are in the same equivalence class also relatively to $S_{[m+2,n]}$ (since it permutes all the elements indexed from $m+2$ to $n$). So, we must have  $v_{m+1}=0$, $v'_{m+1}=1$, and due to properties of the lexicographic array (see Subsection \ref{subsec:lex_array}) we have  $v = u01z$, $v = u10z$, where $|z|=n-m-2$. Since $S_{[m+2,n]}$ fixes the index $m+1$, $v$ and $v'$ are indeed in distinct classes of equivalence. It remains to notice that due to the structure of the lexicographic array there is exactly one such pair of consecutive $v$ and $v'$ which differ at the positions $m$ and $m+1$ for each $m$; so exactly one class is split into two classes. 

Hence for each $t$, $2 \le t \le n + 1$, there exists a subgroup $G$ of $S_n$ for which $p^G(n) = t$. \end{proof}

\section{A necessary condition for universal group complexity and the Thue-Morse word}

In this section we provide a necessary condition for universal group complrxity property. Using this necessary condition, we show that the Thue-Morse word does not have the universal group complexity property.

\begin{proposition}\label{pr:necessary_condition} Let $w$ be a binary word such that its set of factors is closed under the morphism $E: 0\mapsto 1, 1\mapsto 0$ and  that $p_w(n)-a_w(n)\geq 2$ for at least one odd $n$. Then $w$ does not have universal group complexity. Moreover, this holds for each odd length $n$ for which the condition $p_w(n)-a_w(n)\geq 2$ is satisfied.
\end{proposition}

\begin{remark} Note that the condition $p_w(n)-a_w(n)\geq 2$ for at least one odd $n$ is technical and serves to exclude very specific periodic counterexamples like $(01)^\infty$.\end{remark}

\begin{proof} For an odd length $n$, the set of factors of $w$ can be split into pairs of antipodal words, i.e. which are images of one another under the morphism $E$. Consider two factors $v$ and $u$ which are not antipodal, and denote by $v'$ and $u'$ their antipodal words. It is straightforward to see that for each subgroup $G$ of $S_n$ we either have $v\sim_G u$ and $v'\sim_G u'$ or $v\nsim_G u$ and $v'\nsim_G u'$. Thus group complexity cannot have values of the form $p_w(n)-2k-1$ for $k\in \mathbb{N}$. 
\end{proof}

The Thue-Morse word can be defined as the word generated by the morphism  $\mu: 0 \mapsto 01, 1 \mapsto 10$: 
\[
t = 011010011001011010010110\ldots
\]
Since the Thue-Morse word satisfies the conditions from Proposition \ref{pr:necessary_condition}, we have the following:

\begin{corollary} The Thue-Morse word does not have universal group complexity.
\end{corollary}

\section{Ternary words of minimal complexity}
	\ 
	

    In this section, we study the universal group complexity property for ternary words of minimal complexity. In Subsection \ref{subsec:ternary} we show that the universal group complexity property is satisfied for some ternary words of minimal complexity, but not for all of them. In Subsection \ref{subsec:four}, we study separately the value 4 of group complexity, which turns out to be the most tricky case.
	
	\subsection{Ternary words of minimal complexity} \label{subsec:ternary}
	
	Consider the alphabet $\mathbb T = \{0, 1, 2\}$. In \cite{classification} it has been shown that words of minimal complexity on the alphabet $\mathbb T$ (i.e., such that $p_x(n) = n + 2$) could be split into three classes as follows.

 \begin{theorem}[\cite{classification}]\label{th:classification} Let $x$ be an aperiodic word of minimal complexity on $\mathbb T = \{0, 1, 2\}$. Then the set of factors of $x$ coincides (up to renaming letters) with a word from one of the  following classes:
	\begin{enumerate}
		\item $2s$, where $s$ is a Sturmian word;
		\item $\varphi(s)$, where $s$ is a Sturmian word, and $\varphi\colon \begin{cases}
			0 \mapsto 02\\
			1 \mapsto 12
		\end{cases}$;
		\item $\psi(s)$, where $s$ is a Sturmian word, and  $\psi\colon \begin{cases}
			0 \mapsto 0\\
			1 \mapsto 12
		\end{cases}$.
	\end{enumerate}
\end{theorem}

\begin{definition}
    We say that an aperiodic word of minimal complexity on $\mathbb T = \{0, 1, 2\}$ is \emph{of type I} (resp., \emph{type II, type III}) if it is as in item 1. (resp., item 2., item 3.) from Theorem \ref{th:classification}. We say that a word of type III is \emph{of type IIIa} if the word $s$ from Theorem \ref{th:classification} has $0$ as the most frequent letter, otherwise we say it is \emph{of type IIIb}.
\end{definition}



The goal of this section is to prove the following theorem:


\begin{theorem}\label{th:ternary} Let $x$ be an aperiodic word of minimal complexity on $\mathbb T = \{0, 1, 2\}$. 
	\begin{enumerate}
		\item If $x$ is of type I, 
        then $x$ has universal group  complexity.
		\item  If $x$ is of type II, 
        then $x$ does not have universal group complexity. More precisely, the universal group complexity property holds for each length $n$ except for $n=2$.
		\item If $x$ is of type III, 
        then for each length $n$ and each $k$ such that $p_x^{ab}(n)< k < p_x(n)$, $k\neq 4$, there exists a subgroup $G_n^k$ of $S_n$ with $p_x^{G_n^k}=k$. 

	\end{enumerate}
\end{theorem}

\begin{remark} For $x$  of type III, 
the abelian complexity is either 3 or 4 (see Lemma \ref{lm:abelian}). So, the universal group complexity holds for lengths $n$ with $p_x^{ab}(n)=4$. For length $n$ such that $p_x^{ab}(n)= 3$, we do not know precisely when the value 4 is attained by group complexity: there are examples when it is attained and examples when it is not. For particular results on that see Subsection \ref{subsec:four} and Open problem 3. \end{remark}

We now proceed with propositions treating different cases of the theorem and some auxiliary statements, and we finalize the proof in the end of this section.

\begin{proposition}\label{pr:case1}
 Let $x$ be an aperiodic word of minimal complexity of type I. 
 Then $x$ has universal group  complexity.
\end{proposition}
    
\begin{proof}	It is easy to see that the set of factors of $x$ of each length $n$ is given by the set of factors of $s$ of length $n$ plus the prefix of $x$ of length $n$ (beginning with $2$). Clearly, this factor is the only element in its group equivalence class for each subgroup of $S_n$. So, by Theorem \ref{th:sturmian}, we have that for such words the universal group complexity property holds. \end{proof}

	\begin{proposition} \label{pr:case2}
		Let $w$ be a ternary minimal complexity word of type II.
        Then $w$ has universal group complexity for each length $n\geq 3$.  \end{proposition}

In other words, the proposition says that words of this form satisfy the universal group complexity property except for the length $2$.

        
		\begin{proof}
            Let $u$ be a Sturmian word such that $w = \varphi(u)$, where $\varphi\colon \begin{cases}
			0 \mapsto 0\\
			1 \mapsto 12
		\end{cases}$. We consider two cases depending on parity of $n$.
			
			First consider the case when $n = 2k+1$ is odd. The set of factors of $w$ is obtained from the set of factors of $u$ as follows. We have $k + 2$ factors of form $f_12f_22\cdots2f_{k+1}$, where $f_1f_2\cdots f_{k+1}$ is a factor of $u$ of length $k + 1$ and we have $k + 1$ factors of form $2e_12e_2\cdots e_k2$, where $e_1e_2\cdots e_k$ is a factor of $u$ of length $k$. All these factors are distinct, so we have all $2k + 3$ factors of length $2k + 1$ of the word $w$. 
			
			By Theorem \ref{th:sturmian},  we can obtain any integer between $2$ and $k+2$ as group complexity 
            for the word $u$ and length $k + 1$. Let $G\leqslant S_{k+1}$ be the group that gives complexity $p$ for the word $u$. Using $G$, we construct a group $G'\leqslant S_{2k+1}$ acting on factors of $w$ of length $2k+1$ as follows. 
            For each permutation $\sigma\in G$, we define $\sigma'\in G'$ in the following way: if $\sigma(i)=j$, then $\sigma'(2i+1)=2j+1$. In other words, we transfer the action of $G$ to the action of $G'$ on odd indices.
            For example, for the permutation $\sigma=(134)(25)\in G$ we have $\sigma'=(157)(39)$. 
            Symmetrically, we define a subgroup acting on even indices. Let $H\leqslant S_k$ be the group that gives  complexity $q$ for the word $u$. We then define the group $H'\leqslant S_{2k+1}$: for each $\sigma\in H$, we define $\sigma'\in H'$: if $\sigma(i)=j$, then $\sigma'(2i)=2j$. By  definition, $G'$ splits the set of factors of $w$ of length $2k+1$ with $2$'s at even positions into $p$ classes, and does not split the two abelian classes of words with $2$'s at odd positions. Symmetrical situation holds for $H'$. The cartesian product $G' \times H'$ acts on factors of $w$ of length $2k+1$ and gives the complexity $p + q$: factors with different amount of 2's are never group equivalent, moreover, there are $p$ equivalence classes of factors with $k$ occurrences of $2$ and $q$ classes of factors with $k+1$ occurrences of $2$. Hence choosing $p$ and $q$ we can obtain all integers between $ 4$ and $2k + 3$ for the group complexity. It remains to notice that in this case we have abelian complexity equal to $4$. Indeed, there are two classes regarding the number of $2$'s, and for each of these classes there are rich and poor words (as defined for Sturmian words).

			Now consider the case when $n = 2k$ is even. Similarly we can obtain complexities from $4$ to $2k + 2$. In this case the number of $2$'s is the same for all factors of $w$ of length $n$, and we have factors with $2$'s at odd indices and with $2$'s at even indices. The abelian complexity of $w$ is equal to $2$ for the length $n=2k$ as it is equal to the abelian complexity of $u$ for the length $k$. We now show that we can obtain group complexity equal to $3$ with the group that permutes the first $2k-2$ symbols, and independently can switch the last two symbols (the group $S_{[1,n-2]}\times S_{[n-1,n]}$). To show this, we split the factors into at most four classes $P0, P1, R0, R1$ as follows. We decompose a factor of length $n$ into a prefix of length $n-2$ and a suffix of length $2$. For the prefix, there are two abelian classes of factors of length $n - 2$: rich ($R$) and poor ($P$) regarding the number of $1$'s. The last two symbols are $a$ and $2$ (in some order), where $a\in \{0,1\}$. We then call the class $Pa$ or $Ra$ depending on the prefix and the suffix. Since $u$ is balanced, we cannot have both classes $P0$ and $R1$ in $w$, so the number of classes is at most three. On the other hand, at least one of the classes $P0$ and $R1$ is present in $w$, since otherwise we would have the same number of $1$'s for all factors of $u$ and abelian complexity would be $1$.

			Now, we claim that if we have the class $P0$ in $w$, then we also have $P1$. Indeed, in the word $u$, which is Sturmian, for each length, there are two factors from classes $P$ and $R$ which are consecutive 
            in the lexicographic array, otherwise the abelian complexity is 1. So, we have that the last symbol of the factor in $R$ must be $1$. The claim follows.
			Symmetrically, if we have class $R1$, then we have the classes $R0$ and $P1$.
	\end{proof}
    
We note that the universal group complexity property does not hold for the length $n = 2$. Indeed, the factors are $02, 12, 20$ and $21$, so we have $p(2) = 4$ and $p^{ab}(2) = 2$. Since $S_2$ has only two subgroups, we cannot obtain $3$ as group complexity. However, if we unify the letter $2$ with the least frequent letter of the Sturmian word, then we obtain a word satisfying the universal complexity property for each $n$:

   \newpage

    \begin{proposition}
        Let $s$ be a Sturmian word, $a$ be its least frequent letter and $\varphi_a\colon
        \begin{cases}
            0\mapsto0a\\
            1\mapsto1a
        \end{cases}$. Then $\varphi_a(s)$ has universal group complexity.

    \end{proposition}
    \begin{proof}
        Without loss of generality we may assume that $a=1$, as the other case is symmetric.
     
        For $n \leqslant 3$ we have  factors $0, 1; 01, 10, 11; 010, 011, 101, 110, 111$; one can check directly that $\varphi_1(s)$ has universal group complexity for each length $n\leqslant3$.

        For $n \geqslant 4$ we can obtain complexities from $4$ to $n+2$ using groups from Proposition \ref{pr:case2}. 
        Factors with $1$'s at each odd position and with $1$'s at each even postition are distinct as $n \geqslant 4$ and we do not have two consequent $1$'s in $s$. As $\varphi_1(s)$ can be obtained from $\varphi(s)$ by identifying $1$ and $2$, factor complexity of $\varphi_1(s)$ cannot be greater than $n + 2$.

        For even values of $n$ the abelian complexity of $\varphi_1(s)$ is equal to $2$, and we can obtain group complexity $3$ using the same groups as in Proposition \ref{pr:case2}. 

        It remains to prove that we can obtain complexity $3$ in the case of odd values of $n = 2k + 1$ and abelian complexity equal to $2$. Let $l$ be the number of $1$'s in a poor factor of $s$ of length $k$. Then a factor of $\varphi_1(s)$ of length $n$ can have $l + k + 1$ or $l + k + 2$ occurrences of $1$ (as abelian complexity is equal to $2$). So, a factor of $s$ 
        of length $k+1$ can have $l + 1$ or $l + 2$ occurrences of $1$.

        Consider the list of factors of $s$ of length $m$. Suppose that there exists a position $i$ ($1 \leqslant i \leqslant m$) such that all poor factors have $1$ at position $i$ or  all rich factors have $1$ at  position $i$. We will now show that for the word $\varphi_1(s)$ the group complexity equal to $3$ can be achieved for the lengths $2m + 1$ and $2m - 1$ (recall that we now consider only lengths with abelian complexity equal to 2).
        For the length $2m - 1$ we can use the group that permutes letters at all positions except for $2i - 1$ (this group is isomorphic to $S_{2m-2}$). 
       We have three classes of group equivalence relatively to this group: 
        rich 
        factors with $1$ at position $2i - 1$ (as the abelian complexity of $\psi_1(s)$ is 2 for this length, we have two abelian classes of factors, poor and rich factors),  
        poor factors with $1$ at position $2i - 1$, 
        and one class with $0$ at position $2i - 1$. 
        Factors from the latter class have 1's in even positions, and odd positions are filled with  factors of $s$ of length $m$ with $0$ at position $i$. Since we assumed that either every factor of $s$ of length $m$ with $0$ at position $i$ is rich or they all are poor, we have only one group class in $\varphi_1(s)$ with $0$ at position $2i - 1$.
            For the length $2m + 1$ the proof is similar: we can use the group that permutes letters at all positions except for $2i$ (this group is isomorphic to $S_{2m}$). 

        If the right special factor of $s$ of length $m-1$ starts with $1$, we have all rich factors beginning with $1$ and hence we can obtain complexity $3$ for $\varphi_1(s)$ and lengths $2m+1$ and $2m-1$. 
        So, it remains to consider the lengths $n = 2k + 1$ such that right special factors of $s$ of lengths $k - 1$ and $k$ start with $0$. In this case we can obtain complexity 3 using the group $S_{[1, 2]} \times S_{[3, 2k + 1]}$. Poor factors of the word $\varphi_1(s)$ always start with $01$ or $10$, hence they are in the same group equivalence class relatively to the group $S_{[1, 2]} \times S_{[3, 2k + 1]}$. Rich factors can start with $01/10$ or with $11$, so the set of rich factors is split into two distinct group classes. 
    \end{proof}
          
	Let $u$ be an infinite word. For each integer $n$, we consider the lexicographic array of $u$, i.e., the set of factors of $u$ of length $n$ sorted in lexicographic order.
    We now define a \emph{scheme} of $u$ as a set of rules such that for each $n$ each factor (except for the first one) in the lexicographic array can be obtained from the previous one by applying exactly one rule. By a \emph{rule} $AxB \to AyB$, where $x,y\in\Sigma^*$, $|x|=|y|$, we mean that in the word $u$ we can find a factor containing $x$, replace it with $y$, so that the resulting word is another factor  of $u$, and moreover these two factors are consecutive in the lexicographic array.
    A rule of the form $Ax \to Ay$ (resp., $xB \to yB$) can only be applied to the suffix (resp., prefix) $x$. For example, the scheme of any Sturmian word is the set of the following two rules: $A01B \rightarrow A10B$ and $A0 \rightarrow A1$; here $A$ and $B$ are finite words \cite{scheme}. A \emph{$k$-scheme} is a set of rules that holds starting from length $k$.

We now introduce the notion of a \emph{corresponding factor}, i.e., a factor of a Sturmian word $s$ corresponding to a factor of a ternary word of minimal complexity $\psi(s)$: we replace each occurrence of 12 by 1 and, if the factor begins with 2, then we replace this occurrence of 2 with 1. For example, $2012012001 \mapsto 10101001$. 
            \begin{lemma}\label{lm:expansion}
                Let $s$ be a Sturmian word, and $\psi(s)$ a minimal complexity  word, and let $m$, $n$ be integers. Suppose that $X$ and $Y$ are two factors of $\psi(s)$ of length $n$, such that their corresponding factors $\tilde X$ and $\tilde Y$ are of length $m$. If there exists a factor $\tilde Z$ of $s$ of length $m$ such that $\tilde X < \tilde Z < \tilde Y$, then there exists a factor $Z$ of $\psi(s)$ of length $n$ such that $X < Z < Y$.
            \end{lemma}
            \begin{proof}
                We study cases according to the position of the border between factors that start with $1$ and factors that start with $2$ in the lexicographic array of length $n$:
            
                Case 1. $Y$ and hence $X$ do not start with $2$. Then we define $Z$ as a prefix of length $n$ of $\psi(\tilde Z)$. Since $\tilde Z > \tilde X$, the factor $\tilde{Z}$ cannot have less $1$'s than $\tilde X$ (see Subsection \ref{subsec:lex_array} on lexicographic arrays of Stumian words), so $|\psi(\tilde{Z})| \geqslant n$. 
                Note that the morphism $\psi$ preserves the lexicographic order, i.e., if for some words $u$ and $v$ we have $u<v$, then $\psi(u)<\psi(v)$. 
                So, we have $\psi(\tilde X) < \psi(\tilde Z) < \psi(\tilde Y)$. Taking prefixes of these words of length $n$, we obtain $X \leqslant Z \leqslant Y$. 
                
                Now we have to prove that the equality is not possible, i.e., we cannot have $Z=X$ or $Z=Y$. In Case 1 we remove no more than one letter, i.e., $\psi(\tilde X)=X$ or $\psi(\tilde X)=Xa$, and the same holds for $\psi(\tilde Z)$. So, if $X = Z$, then $\psi(\tilde X) < \psi(\tilde Z)$ implies either $X < Zb$ or $Xa < Zb$. Note that for the morphism $\psi$, the last letter of $\psi(w)$ is always determined by the penultimate letter: if the penultimate letter is $0$ or $2$, then the last one is $0$, and if the penultimate letter is $1$, then the last one is $2$. This implies that the case $Xa < Zb$ is impossible. If we have $X < Zb$, then $b = 0$: $Zb$ cannot end with $1$, and if $b = 2$, then $Z = X$ ends with $1$, which is impossible as well. So, the remaining case is $\psi(\tilde{X})=X<\psi(\tilde{Z})=Z0$, hence $\tilde Z = \tilde X0$. We reach a contradiction since $\tilde Z$ and $\tilde X$ must be of the same length by the conditions of the lemma. So, we proved that $Z\neq X$. The proof of $Z \neq Y$ is symmetric.

                Case 2. $X$ and hence $Y$ starts with $2$. This case is reduced to Case 1 as follows. Consider factors $1X$ and $1Y$ satisfying Case 1. We proved that there exists a factor $Z'$ such that $1X < Z' < 1Y$, which means that $Z'$ starts with $1$. If we remove it, we obtain the needed factor $Z$.

                Case 3. $X$ starts with $1$, $Y$ starts with $2$. We can assume that $\tilde Z$ is the factor following $\tilde X$ factor in the lexicographic array for  $s$. It can be obtained applying rule $A01B \rightarrow A10B$ or $A0 \rightarrow A1$ to $\tilde{X}$. In either case, we define $Z$ as a prefix of length $n$ of $\psi(\tilde{Z})$ and we have $Z > \psi(\tilde X)$; moreover, $Z$  starts with $1$, so $Z < Y$.
            \end{proof}
            
            
            
            
	
	The following theorem provides a scheme for a ternary minimal complexity word of type III. 

    
	\begin{theorem}\label{th:scheme} 
    A ternary minimal complexity word of type III has a $2$-scheme consisting of the following four rules: 
 \begin{enumerate}   
    \item $A012B\rightarrow A120B$,
    
    \item $A0 \rightarrow A1$, 
    
    \item $A01 \rightarrow A12$, 
    
    \item $12B \rightarrow 20B$.
\end{enumerate}    

     Moreover, each of the rules 2--4 is applied exactly once for each length.
    \end{theorem}
		\begin{proof} We let $n$ denote the length of the lexicographic array.
We start with the second statement. Rule 2 is applied when $A$ is a left special factor of length $n-1$. As in a minimal complexity word of type III we have a unique right special factor of each length, which is extended to the right by $0$ and $1$, Rule 2 is applied exactly once. Similarly for Rule 3 with $A$ being right special of 
length $n-2$, and for Rule 4 with $B$ being left special factor of length $n-2$.
        
			We now prove the first statement. We write $X \vdash Y$ for factors $X$ and $Y$ if in the lexicographic array the word $X$ is followed by the word $Y$, and $Y$ obtained from $X$ via one of the four rules from the statement of the theorem.

           First consider $n=2$. For words of type IIIa we have four factors $00 \vdash 01 \vdash 12 \vdash 20$. The applied rules are 2, 3 and 4, respectively. For words of type IIIb we have four factors $01 \vdash 12 \vdash 20 \vdash 21$. The applied rules are 3, 4 and 2, respectively.
			
			  The proof for length $n\geqslant3$ is by induction on length. 
			
			For the base $n = 3$ and type IIIa, we have two possibilities for the list of factors, depending on whether the word contains the factor $000$ or not:
            
            \begin{itemize} 
            
            \item  If $000$ is not a factor: $001 \vdash 012 \vdash 120 \vdash 200 \vdash 201$. The applied rules are 3, 1, 4 and 2, respectively.
			
			\item If $000$ is a factor: $000 \vdash 001 \vdash 012 \vdash 120 \vdash 200$. The applied rules 2, 3, 1 and 4, respectively.

            \end{itemize}

            For type IIIb, we have $012 \vdash 120 \vdash 121 \vdash 201 \vdash 212$. The applied rules are 1, 2, 4 and 3, respectively.

			We proceed with the induction step $n \mapsto n+1$. Consider two consecutive factors of length $n + 1$: $Xu$ and $Yv$, where $X$ and $Y$ are factors of length $n$, and $u$ and $v$ are letters.
			
			\begin{center}
				\begin{tikzpicture}[scale=0.6]
					\draw[black, very thick] (0,0) rectangle (6,1);
					\node [anchor=center] at (3, 0.5) {$Y$};
					\draw[red, very thick] (6,0) rectangle (7,1);
					\node [anchor=center] at (6.5, 0.5) {$v$};
					\draw[black, very thick] (0,1) rectangle (6,2);
					\node [anchor=center] at (3, 1.5) {$X$};
					\draw[red, very thick] (6,1) rectangle (7,2);
					\node [anchor=center] at (6.5, 1.5) {$u$};
				\end{tikzpicture}
			\end{center}
			
			Case I: $X = Y$. Recall that for a word of minimal complexity we have exactly one right special factor of length $n$ which can be extended to the right in two different ways. 
            For a word of type IIIa, after each occurrence of 1 in both words we have an occurrence of 2, and after an occurrence of 2 we have an occurrence of 0, hence our right special factor ends with $0$ and can be continued by 0 or 1, therefore $u = 0, v = 1$, which corresponds to rule 2. For a word of type IIIb, we similarly have $u = 0, v = 1$, as an occurrence of $1$ is followed by an occurrence of $2$ and an occurrence of $0$ is followed by an occurrence of $1$.
			
			Case II: $X \neq Y$. If $X$ and $Y$ are not consecutive in the lexicographic order, then there is at least one factor between $Xu$ and $Yv$, which is false, so $Y$ follows $X$ in the lexicographic array of length $n$.
			
			Case II-1: $u = v$.
			
			If $u = v = 0$, then, since $X \vdash Y$, we have that $Y$ is obtained from $X$ by applying rules 1 or 4 (neither $X$ nor $Y$ can end with 1, since each occurrence of 1 is followed by an occurrence of 2). Then we have $X0 \vdash Y0$ with the same rule.
			
			\begin{center}
				\begin{tikzpicture}[scale=0.6]
					\draw[black, very thick] (0,0) rectangle (6,1);
					\node [anchor=center] at (3, 0.5) {$Y$};
					\draw[red, very thick] (6,0) rectangle (7,1);
					\node [anchor=center] at (6.5, 0.5) {$0$};
					\draw[black, very thick] (0,1) rectangle (6,2);
					\node [anchor=center] at (3, 1.5) {$X$};
					\draw[red, very thick] (6,1) rectangle (7,2);
					\node [anchor=center] at (6.5, 1.5) {$0$};
				\end{tikzpicture}
			\end{center}
			
			If $u = v = 1$, then, since occurrences of 1's are preceded by occurrences of 0's 
            or 2's, we again can apply only rules $1$ and $4$ for $X$ and $Y$.
			
			\begin{center}
				\begin{tikzpicture}[scale=0.6]
					\draw[blue, thick] (5,0) rectangle (6,2);
					\node [anchor=center] at (5.5, 0.5) {$0/2$};
					\node [anchor=center] at (5.5, 1.5) {$0/2$};
					
					\draw[black, very thick] (0,0) rectangle (6,1);
					\node [anchor=center] at (0.5, 0.5) {$Y$};
					\draw[red, very thick] (6,0) rectangle (7,1);
					\node [anchor=center] at (6.5, 0.5) {$1$};
					\draw[black, very thick] (0,1) rectangle (6,2);
					\node [anchor=center] at (0.5, 1.5) {$X$};
					\draw[red, very thick] (6,1) rectangle (7,2);
					\node [anchor=center] at (6.5, 1.5) {$1$};
				\end{tikzpicture}
			\end{center}
			
			In the case $u = v = 2$ the proof is symmetric.
			
			\begin{center}
				\begin{tikzpicture}[scale=0.6]
					\draw[blue, thick] (5,0) rectangle (6,2);
					\node [anchor=center] at (5.5, 0.5) {$1$};
					\node [anchor=center] at (5.5, 1.5) {$1$};
					
					\draw[black, very thick] (0,0) rectangle (6,1);
					\node [anchor=center] at (0.5, 0.5) {$Y$};
					\draw[red, very thick] (6,0) rectangle (7,1);
					\node [anchor=center] at (6.5, 0.5) {$2$};
					\draw[black, very thick] (0,1) rectangle (6,2);
					\node [anchor=center] at (0.5, 1.5) {$X$};
					\draw[red, very thick] (6,1) rectangle (7,2);
					\node [anchor=center] at (6.5, 1.5) {$2$};
				\end{tikzpicture}
			\end{center}
			
			Case II-2: $u \neq v$.
			

            If $u=1$ and $v=2$, this occurrence of $1$ is preceded by $0$ or $2$ and the occurrence of $2$ is preceded by 1.
			
			\begin{center}
				\begin{tikzpicture}[scale=0.6]
					\draw[blue, thick] (5,0) rectangle (6,2);
					\node [anchor=center] at (5.5, 0.5) {$1$};
					\node [anchor=center] at (5.5, 1.5) {$0/2$};
					
					\draw[black, very thick] (0,0) rectangle (6,1);
					\node [anchor=center] at (0.5, 0.5) {$Y$};
					\draw[red, very thick] (6,0) rectangle (7,1);
					\node [anchor=center] at (6.5, 0.5) {$2$};
					\draw[black, very thick] (0,1) rectangle (6,2);
					\node [anchor=center] at (0.5, 1.5) {$X$};
					\draw[red, very thick] (6,1) rectangle (7,2);
					\node [anchor=center] at (6.5, 1.5) {$1$};
				\end{tikzpicture}
			\end{center}
			
			We have that $X \vdash Y$; the only rule that can be applied in this case is rule 2. 
            So, we have $X1 \vdash Y2$ using rule 3.
			
			If $u=2$ and $v=0$, then 2 is preceded by 1 
            , and 0 is preceded by either 0 or 2: 
			
			\begin{center}
				\begin{tikzpicture}[scale=0.6]
					\draw[blue, thick] (5,0) rectangle (6,2);
					\node [anchor=center] at (5.5, 0.5) {$0/2$};
					\node [anchor=center] at (5.5, 1.5) {$1$};
					
					\draw[black, very thick] (0,0) rectangle (6,1);
					\node [anchor=center] at (0.5, 0.5) {$Y$};
					\draw[red, very thick] (6,0) rectangle (7,1);
					\node [anchor=center] at (6.5, 0.5) {$0$};
					\draw[black, very thick] (0,1) rectangle (6,2);
					\node [anchor=center] at (0.5, 1.5) {$X$};
					\draw[red, very thick] (6,1) rectangle (7,2);
					\node [anchor=center] at (6.5, 1.5) {$2$};
				\end{tikzpicture}
			\end{center}
			
			Then $X \vdash Y$ using rule 3, and $X2 \vdash Y0$ using rule 1.
			
			We will now prove that the case $u=0$ and $v=1$ is impossible. 

            In this case, we have that $X$ and $Y$ end with 0 or 2:

            \begin{center}
				\begin{tikzpicture}[scale=0.6]
					\draw[blue, thick] (5,0) rectangle (6,2);
					\node [anchor=center] at (5.5, 1.5) {$0/2$};
					
					\node [anchor=center] at (5.5, 0.5) {$0/2$};
					
					\draw[black, very thick] (0,0) rectangle (6,1);
					\node [anchor=center] at (0.5, 0.5) {$Y$};
					\draw[red, very thick] (6,0) rectangle (7,1);
					\node [anchor=center] at (6.5, 0.5) {$1$};
					\draw[black, very thick] (0,1) rectangle (6,2);
					\node [anchor=center] at (0.5, 1.5) {$X$};
					\draw[red, very thick] (6,1) rectangle (7,2);
					\node [anchor=center] at (6.5, 1.5) {$0$};
				\end{tikzpicture}
			\end{center}

            Hence we have $X \vdash Y$ using rule 1 or 4. If this is rule 1, then $\tilde{X}$ and $\tilde{Y}$ are of the same length, but are not equal, so in the lexicographic array of $s$ between $\tilde{X}0$ and $\tilde{Y}1$ there is a factor $\tilde{Z}$ of $s$. Applying Lemma \ref{lm:expansion}, we obtain a factor of $\psi(s)$ between $X0$ and $Y1$ in the lexicographic array of $\psi(s)$. A contradiction with the initial assumption that $X0$ and $Y1$ are consecutive factors in the lexicographic array for length $n+1$.

            If this is rule $4$, then we consider a picture: 

            \begin{center}
				\begin{tikzpicture}[scale=0.6]
					\draw[blue, thick] (0,0) rectangle (1,2);
					\draw[blue, thick] (1,0) rectangle (2,2);
					\node [anchor=center] at (0.5, 1.5) {$1$};
					\node [anchor=center] at (1.5, 1.5) {$2$};
					\node [anchor=center] at (0.5, 0.5) {$2$};
					\node [anchor=center] at (1.5, 0.5) {$0$};

					\draw[black, very thick] (0,0) rectangle (6,1);
					\node [anchor=center] at (4, 0.5) {$A$};
					\draw[red, very thick] (6,0) rectangle (7,1);
					\node [anchor=center] at (6.5, 0.5) {$1$};
					\draw[black, very thick] (0,1) rectangle (6,2);
					\node [anchor=center] at (4, 1.5) {$A$};
					\draw[red, very thick] (6,1) rectangle (7,2);
					\node [anchor=center] at (6.5, 1.5) {$0$};
				\end{tikzpicture}
			\end{center}

            The factor $\tilde{A}$ of the word $s$ is right special, so, as Sturmian words have exactly one right special factor of each length, we have that either $0\tilde A$ or $1 \tilde A$ is right special as well.
            
            In the first case we have  that $0\tilde A 0$ is a factor of $s$ and $0A0$ is a factor of $\psi(s)$. For a word of type IIIa, we can extend both of $X0$ and $Y1$ to the left (extensions are uniquely defined) and obtain $012A0$ and $120A1$. Their corresponding factors are $01\tilde A 0$ and $10\tilde A 1$, respectively. These factors are not consecutive in the lexicographic array of the word $s$, so we have a factor in between. By Lemma \ref{lm:expansion} we can obtain a factor $aZ$ (where $a$ is a letter and $Z$ is a factor) between $0X0$ and $1Y1$ in the lexicographic array of $\psi(s)$. If $a = 0$, we have $X0 < Z$, and since $Y$ starts with 2 and $Z$ does not, we have $Z < Y1$. 
            If $a = 1$, we have $Z < Y1$ and since $Z$ starts with 2 and $X$ does not, we have $X0 < Z$. 
            For a word of type IIIb, the factor $0A0$ is extended to the left in a unique way by 2, so we obtain a factor $20A0$. In each of the cases we obtain a contradiction with the initial assumption that $X0$ and $Y1$ are consecutive factors in the lexicographic array for length $n+1$. 
            
            In the case if $1 \tilde A$ is a right special factor, we have that $1 \tilde A 1$ is a factor of $s$ and hence $12A1$ is a factor of $\psi(s)$; this again contradicts the initial assumption that $X0$ and $Y1$ are consecutive factors in the lexicographic array for length $n+1$.

			Now consider the case $u=0$, $v=2$. For words of type IIIb it means that $X$ ends with $2$ and $Y$ ends with $1$, and we do not have a rule for $X \vdash Y$. So, we only need to consider words of type IIIa in this case. For a word of type IIIa the factor $Y$ must end with $01$, and the word $X$ cannot end with $1$. So, we can only apply rule 2 for $X\vdash Y$. Then $X$ ends with $00$. Consider the extensions of these factors to the right:
			
			\begin{center}
				\begin{tikzpicture}[scale=0.6]
					\draw[blue, thick] (4,0) rectangle (5,2);
					\draw[blue, thick] (5,0) rectangle (6,2);
					\node [anchor=center] at (4.5, 0.5) {$0$};
					\node [anchor=center] at (4.5, 1.5) {$0$};
					\node [anchor=center] at (5.5, 0.5) {$1$};
					\node [anchor=center] at (5.5, 1.5) {$0$};
					
					\draw[black, very thick] (0,0) rectangle (6,1);
					\node [anchor=center] at (0.5, 0.5) {$Y$};
					\draw[red, very thick] (6,0) rectangle (7,1);
					\node [anchor=center] at (6.5, 0.5) {$2$};
					\draw[black, very thick] (0,1) rectangle (6,2);
					\node [anchor=center] at (0.5, 1.5) {$X$};
					\draw[red, very thick] (6,1) rectangle (7,2);
					\node [anchor=center] at (6.5, 1.5) {$0$};
					
					\draw[green, very thick] (7,0) rectangle (8,1);
					\node [anchor=center] at (7.5, 0.5) {$0$};
					\draw[green, very thick] (7,1) rectangle (8,2);
					\node [anchor=center] at (7.5, 1.5) {$?$};
				\end{tikzpicture}
			\end{center}
			
			Note that there are two corresponding factors of the same length: $\tilde{A}000$ and $\tilde{A}010$, where $A$ is a prefix of $X$ and $Y$ defined by $X = A00, Y = A01$. Factors $\tilde{A}000$ and $\tilde{A}010$ are not consecutive in the lexicographic array of a Sturmian word, so there is a factor between them, which is $\tilde{A}001$. This factor is a corresponding factor to $X1$; a contradiction with the fact that $X0$ and $Y2$ are consecutive factors in the lexicographic array for length $n+1$.
			
			Now consider the case $u=1$, $v=0$. Then $X$ and $Y$ end either with $0$ or with $2$. 
			
			\begin{center}
				\begin{tikzpicture}[scale=0.6]
					\draw[blue, thick] (5,0) rectangle (6,2);
					\node [anchor=center] at (5.5, 0.5) {$0/2$};
					\node [anchor=center] at (5.5, 1.5) {$0/2$};
					
					\draw[black, very thick] (0,0) rectangle (6,1);
					\node [anchor=center] at (0.5, 0.5) {$Y$};
					\draw[red, very thick] (6,0) rectangle (7,1);
					\node [anchor=center] at (6.5, 0.5) {$0$};
					\draw[black, very thick] (0,1) rectangle (6,2);
					\node [anchor=center] at (0.5, 1.5) {$X$};
					\draw[red, very thick] (6,1) rectangle (7,2);
					\node [anchor=center] at (6.5, 1.5) {$1$};
				\end{tikzpicture}
			\end{center}
			
			We have $X \vdash Y$ using rule 1 or rule 4. If we have $X \vdash Y$ using rule 1, then $\tilde{X}1$ and $\tilde{Y}0$ are of the same length, and $\tilde{X}1$ is smaller in the lexicographic order. However, it contains more 1's; a contradiction. 		If we have $X \vdash Y$ using rule 4, then consider the following picture: 
            			
			\begin{center}
				\begin{tikzpicture}[scale=0.6]
					\draw[blue, thick] (0,0) rectangle (1,2);
					\draw[blue, thick] (1,0) rectangle (2,2);
					\node [anchor=center] at (0.5, 1.5) {$1$};
					\node [anchor=center] at (1.5, 1.5) {$2$};
					\node [anchor=center] at (0.5, 0.5) {$2$};
					\node [anchor=center] at (1.5, 0.5) {$0$};
					
					\draw[black, very thick] (0,0) rectangle (6,1);
					\node [anchor=center] at (4, 0.5) {$A$};
					\draw[red, very thick] (6,0) rectangle (7,1);
					\node [anchor=center] at (6.5, 0.5) {$0$};
					\draw[black, very thick] (0,1) rectangle (6,2);
					\node [anchor=center] at (4, 1.5) {$A$};
					\draw[red, very thick] (6,1) rectangle (7,2);
					\node [anchor=center] at (6.5, 1.5) {$1$};
				\end{tikzpicture}
			\end{center}
			The corresponding factors for $X1$ and $Y0$ are $1\tilde{A}1$ and $10\tilde{A}0$, respectively. So,  $1\tilde{A}1$ and $0\tilde{A}0$ are factors of a balanced word $s$, which is impossible. 			
			
			Now we consider the case $u=2$ and $v=1$. Then $X$ ends with $1$ and $Y$ ends with 0 or 2.
			
			\begin{center}
				\begin{tikzpicture}[scale=0.6]
					\draw[blue, thick] (5,0) rectangle (6,2);
					\node [anchor=center] at (5.5, 0.5) {$0/2$};
					\node [anchor=center] at (5.5, 1.5) {$1$};
					
					\draw[black, very thick] (0,0) rectangle (6,1);
					\node [anchor=center] at (0.5, 0.5) {$Y$};
					\draw[red, very thick] (6,0) rectangle (7,1);
					\node [anchor=center] at (6.5, 0.5) {$1$};
					\draw[black, very thick] (0,1) rectangle (6,2);
					\node [anchor=center] at (0.5, 1.5) {$X$};
					\draw[red, very thick] (6,1) rectangle (7,2);
					\node [anchor=center] at (6.5, 1.5) {$2$};
				\end{tikzpicture}
			\end{center}
            
            We have $X \vdash Y$ using rule 3, i.e., $X = A01$, $Y=A12$. Then the corresponding factors for $X2$ and $Y1$ are $\tilde A 01$ and $\tilde A 11$, respectively. They are not consecutive factors in the lexicographic array of the word $s$, so we can apply Lemma \ref{lm:expansion} to a factor between them and obtain a factor between $X2$ and $Y1$. A contradiction with the assumption that $X0$ and $Y1$ are consecutive factors in the lexicographic array for length $n+1$.
            
			Here we provide the tree for the case study for the proof of the theorem:

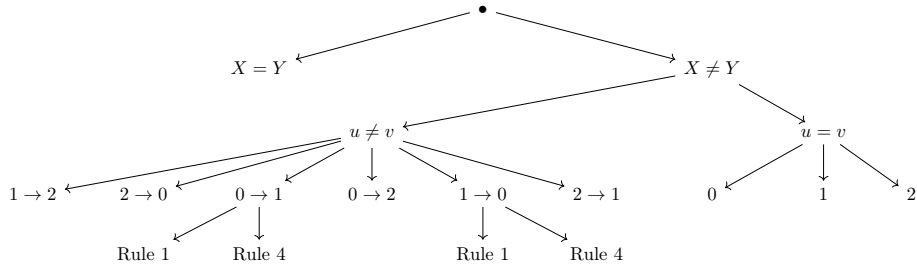
\[
\scalebox{0.7}{
\begin{tikzcd}[ampersand replacement=\&]
	\&\&\&\& \bullet \&\&\&\& \\
	\&\& {X = Y} \&\&\&\& {X\neq Y} \\
	\&\&\& {u\neq v} \&\&\&\& {u=v} \\
	{1\rightarrow2} \& {2\rightarrow0} \& {0\rightarrow1} \& {0\rightarrow2} \& {1\rightarrow0} \& {2\rightarrow1} \& 0 \& 1 \& 2 \\
	\& {\textrm{Rule 1}} \& {\textrm{Rule 4}} \&\& {\textrm{Rule 1}} \& {\textrm{Rule 4}}
	\arrow[from=1-5, to=2-3]
	\arrow[from=1-5, to=2-7]
	\arrow[from=2-7, to=3-4]
	\arrow[from=2-7, to=3-8]
	\arrow[from=3-4, to=4-1]
	\arrow[from=3-4, to=4-2]
	\arrow[from=3-4, to=4-3]
	\arrow[from=3-4, to=4-4]
	\arrow[from=3-4, to=4-5]
	\arrow[from=3-4, to=4-6]
	\arrow[from=3-8, to=4-7]
	\arrow[from=3-8, to=4-8]
	\arrow[from=3-8, to=4-9]
	\arrow[from=4-3, to=5-2]
	\arrow[from=4-3, to=5-3]
	\arrow[from=4-5, to=5-5]
	\arrow[from=4-5, to=5-6]
\end{tikzcd}
}
\]

			
		\end{proof}

    
We now prove the following lemma about the structure of the last column of the lexicographic array of a word of Type III:
    
	\begin{lemma} \label{lm:four}


        The last column of the lexicographic array of a ternary minimal complexity word of type III of each length has one of the following forms: $0^+1^+2^+0^+$, $1^+2^+0^+1^+$ or $2^+0^+1^+2^+$.
    \end{lemma}
		
    \begin{proof} 
			
			The proof is by induction on the length $n$ of the lexicographic array.
			
			Base case is $n = 3$; by a straightforward case study we get that the last column is either $12001$ or $01200$ for words of type IIIa and $20112$ for words of type IIIb.

			For brevity, we will denote a block of the same letter with this letter (\emph{reduced column}).
			
			Suppose that the statement holds for length $n$, i.e., the reduced last column is $0120$, $1201$ or $2012$. 
            We are going to prove that the reduced last column for length $n+1$ is also of this form.
			
			First suppose that the reduced column is of the form $0120$. Consider the extension of  all factors by one letter on the right. Note that 1 must be followed by 2, 2 and 0 must be followed by either 0 or 1. By Theorem \ref{th:scheme}, if we have in the lexicographic array two consecutive factors $X\vdash Y$ and they do not end with the same letter, then $X$ ends with $1$ (resp., $2$, $0$) and $Y$ ends with $2$ (resp., $0$, $1$). In addition, Theorem \ref{th:scheme} implies that any for any pair of two consecutive factors of the form $X0\vdash Y1$ we have that $X=Y$ is a right special factor. So, since we have exactly one right special factor of each length $n$ (extended by $0$ and $1$), the reduced last column has the form either $0120$, or $1201$, according to the location of the right special factor. 
			
			In the remaining cases the proof is symmetric. In the case when the reduced last column for length $n$ is $1201$, for length $n+1$ we obtain 
            $2012$. 
			In the case  when the reduced last column for length $n$ is $2012$, for length $n+1$ we obtain $0120$ or $1201$, according to the location of the right special factor.
	\end{proof}

Since each column of a lexicographic array can be considered as the last column of a smaller array (with some factors repeated), we have the following corollary:
    	
	\begin{corollary} \label{cor:each_col}
		Each column of the lexicographic array of a ternary minimal complexity word of type III 
        of each length 
        has one of the following forms: $0^+1^+2^+0^+$, $1^+2^+0^+1^+$ or $2^+0^+1^+2^+$.
    \end{corollary}

    We recall that the subgroup $S_{[k, n]}$ of $S_n$ consists of all permutations of all elements from $k$ to $n$.
	
	\begin{lemma} \label{lm:five}
		The group complexity of a ternary minimal complexity word of type III of length $n$ relatively to  the subgroup $S_{[2, n]}$ is at least $5$. \end{lemma}
		\begin{proof}
			We will show that the lexicographic array can be split into five classes of lexicographically consecutive factors such that factors from different classes are not group equivalent.
         		
			First we note that splitting into three classes is straightforward: words starting with different letters cannot be in the same group equivalence class for the group $S_{[2,n]}$.

            By Theorem \ref{th:scheme}, in the lexicographic array of length $n$ each of the rules 2--4 is applied exactly once.  
                    
			Now, consider the class containing two factors that have the (unique) right special factor $A$ 
            of length $(n-1)$ as a prefix, i.e., factors $A0$ and $A1$. These two factors are not abelian equivalent (since they contain different numbers of occurrences of 1's), so this class is further split into two classes that are not group equivalent. By Theorem \ref{th:scheme}, the number of 1's is changed only when applying rules 2 and 4. Here we apply rule 2, and rule 4 changes the first letter. So, the factors which are lexicographically smaller than $A0$ and start with the same letter as $A$ are not abelian equivalent to those which are lexicographically bigger than $A1$ and start with the same letter as $A$.
            
			Finally, when we apply rule 3 to a factor, the obtained word is not equivalent to initial factor, so this class is also split into two classes (with different numbers of occurrences of 2's). 
            With a symmetric argument as above (with rule 3 instead rules 2, 4 and number of 2's instead of number of 1's) we have that factors which are lexicographically smaller than $A01$ and start with the same letter as $A$ are not abelian equivalent to those which are lexicographically bigger than $A12$ and start with the same letter as $A$. 
	\end{proof}


    \begin{example}\label{ex:5psif1} As an example for Lemma \ref{lm:five}, consider the word $\psi(f)$, where $f$ is the Fibonacci word and $\psi\colon \begin{cases}
			0 \mapsto 0\\
			1 \mapsto 12
		\end{cases}.$ Its lexicographic array of factors of length $5$ is split into five classes; we illustrate the borders between the classes with horizontal lines:
            \begin{center}
				\begin{tabular}{lllll}
					0 & 0 & 1 & 2 & 0 \\ 
					0 & 1 & 2 & 0 & 0 \\ \hline
					0 & 1 & 2 & 0 & 1 \\ \hline
					1 & 2 & 0 & 0 & 1 \\ \hline
					1 & 2 & 0 & 1 & 2 \\ \hline
					2 & 0 & 0 & 1 & 2 \\ 
					2 & 0 & 1 & 2 & 0 \\ 
				\end{tabular}
			\end{center}
\end{example}
    
	\begin{proposition}\label{pr:no-four} For each length $n$ and for each integer $k$, $2\leqslant k \leqslant n - 1$, the group complexity of a ternary minimal complexity word of type III relatively to the subgroup $S_{[k, n]}$ is equal to $k + 3$. In particular, the values $5, \dots, n + 2$ are achieved by group complexity. \end{proposition}
		\begin{proof}
			We will prove that for the group $S_{[k+1, n]}$ the complexity is bigger than for the group  $S_{[k, n]}$, for $2 \leqslant k \leqslant n - 2$.
			
			We let $A$ denote the (unique) right special factor of length  $k - 1$. This factor is extended to the right either as $A0\cdots$ or as $A1\cdots$. 
            
            We let $X$ (resp., $Y$) denote the lexicographically largest (resp., smallest) factor beginning with $A0$ (resp., $A1$). 
            In the lexicographic array $X$  is followed by $Y$, hence $X \vdash Y$. We have $k - 1 \leqslant n - 3$, i.e., in $X$ after the prefix $A$ there are at least three symbols, so rule 1 ($A012B \vdash A120B$) is applied. So, these factors are in the same group complexity class for $S_{[k, n]}$ and in different classes for $S_{[k+1, n]}$, as they have different prefixes of length $k$.
								
			By Lemma \ref{lm:five}, for $S_{[2, n]}$ the complexity is at least $5$; for $S_{[k, n]}$ it is at least $k + 3$ for $k\leqslant n - 1$. All values are distinct for distinct $k$, and the upper bound is given by factor complexity $n + 2$, which corresponds to the group complexity for $k = n - 1$. Thus the complexity for $S_{[k, n]}$ is equal to $k + 3$ for each $k$ between $2$ and $n - 1$.
	\end{proof}

	\begin{lemma}\label{lm:abelian}
		The abelian complexity of a ternary minimal complexity word 
        of type III is equal to 3 or 4. \end{lemma}
		
	\begin{proof}
			Consider the lexicographic array of factors of length $n$ and rules which are used for obtaining a word in each row from the previous one.

            
            Among the four rules of the scheme for a word of type III only rules 2--4 change the abelian complexity. Applying rule 3 changes the number of 2's and hence changes the abelian class. Applying rule 4 changes the number of 1's and hence changes the abelian class as well. Thus we have at least three abelian classes. 
            Since each of these rules is applied once, we have at most four abelian classes. 
	\end{proof}


	\subsection{Group complexity 4 in words of minimal complexity over ternary alphabet}\label{subsec:four}

     Theorem \ref{th:ternary} gives the classification of  group complexity values except for value 4 for the words as in item 3 of Theorem \ref{th:classification} and for lengths with abelian complexity 3. In this subsection we find a criterion for abelian complexity equal to 4 (Proposition \ref{pr:abelian-four-ends-zero}). We also show in Theorem \ref{th:transitive} that there are infinitely many lengths for which group complexity 4 cannot be attained for non-transitive subgroups of $S_n$.
     

   We now analyze in which order the rules in the scheme of a ternary minimal complexity word of type III can be used for the lengths with abelian complexity equal to 4. 
    
	\begin{lemma}\label{lm:abelian-rules}
		The abelian complexity of a ternary minimal complexity word of type III is equal to $4$ for some length $n$ if and only if rule 3 is applied between rules 2 and 4.
    \end{lemma}
		
	\begin{proof} First note that by Theorem \ref{th:scheme}, for each length rules 2--4 are applied exactly once, and applying rule 1 does not change the abelian class of the factor.

            If rule 3 is applied between rules 2 and 4, then, since it splits the set of factors into two abelian classes with distinct number of 2's, we have four abelian classes in total: each of the remaining rules splits each of the two classes into two classes with distinct numbers of 0's and 1's, and all the four classes are distinct as well.
            
			If the order is 2--4--3: First we replace one occurrence of 0 with 1 (rule 2) getting the second abelian class, then replace one occurrence of 1 with 0 (rule 4), getting again the first abelian class, and finally we replace one occurrence of 2 with 1 (rule 3), getting the third abelian class. So, with this order we  have three abelian classes.

			If the order is 3--4--2, we obtain three classes similarly to the above: first we add 2 and remove 0, then we add 0 and remove 1, and finally add 1 and remove 0. 
            
            In the remaining two cases (3--2--4 and 4--2--3) similarly to the above we get three abelian classes.
\end{proof}

    For a word $w$ and a integer $n$ we let $F_n(w)$ denote the set of factors $w$ of length $n$.

    \begin{lemma}\label{lm:eta}

Let $\eta: \mathbb{T}^* \to \mathbb{T}^*$ be a map defined as a composition $\chi \circ \rho$, where $\rho$ is the reversal map and $\chi$ is the morphism $\begin{cases}
            0 \mapsto 0\\
            1 \mapsto 2\\
            2 \mapsto 1
        \end{cases}$. Then for each $n$ and for each ternary minimal complexity word $w$ of type III the map $\eta$ is an involution (in the sense that $\eta^2=Id$) on the set $F_n(w)$.
    \end{lemma}

    \begin{proof}
        The claim follows immediately from the fact that the set of factors of a Sturmian word is closed under reversal.    
%
%
%
%
    \end{proof}




    \begin{remark}
        For a ternary minimal complexity word $w$ of type III, we have the following relation for the numbers of factors ending with letter 1 and ending with letter 2  for each $n \geqslant 1$: $$|F_{n+1}(w) \cap \mathbb T^*2| = |F_{n}(w) \cap \mathbb T^*1|.$$ This follows from the fact that every factor from $F_{n}(w) \cap \mathbb T^*1$ can be continued only with $2$, and each factor $F_{n+1}(w) \cap \mathbb T^*2$ has $1$ as the second last symbol. 
    \end{remark}

    The following lemma gives some further relations between the numbers of factors beginning or ending with letters 1 and 2:

    \begin{lemma}\label{lm:edges}
        Let $w$ be a ternary minimal complexity word of type III. For each $n\geqslant 1$ the following equalities hold: $$|F_{n}(w) \cap 1\mathbb T^*| = |F_{n}(w) \cap \mathbb T^*2|,$$ $$|F_{n}(w) \cap 2\mathbb T^*| = |F_{n}(w) \cap \mathbb T^*1|,$$ $$|F_{n}(w) \cap 0\mathbb T^*| = |F_{n}(w) \cap \mathbb T^*0|.$$
    \end{lemma}
    \begin{proof} For $n = 1$ the claim is obvious. Consider $n \geqslant 2$.
        
        The involution $\eta$ from Lemma \ref{lm:eta} maps factors from the set $F_{n}(w) \cap 1\mathbb T^*0$ 
        to the set $F_{n}(w) \cap 0\mathbb T^* 2$ and factors from $F_{n}(w) \cap 1\mathbb T^*1$ to $F_{n}(w) \cap 2\mathbb T^* 2$. Hence $$|F_{n}(w) \cap 1\mathbb T^*| = |F_{n}(w) \cap 1\mathbb T^*0| + |F_{n}(w) \cap 1\mathbb T^* 1| + |F_{n}(w) \cap 1\mathbb T^* 2| = $$ $$ |F_{n}(w) \cap 0\mathbb T^* 2| + |F_{n}(w) \cap 2\mathbb T^* 2| + |F_{n}(w) \cap 1 \mathbb T^*2| = |F_{n}(w) \cap\mathbb T^*2|.$$ This proves the first part.

        The second part can be proved in a similar way using the equalities $$|F_{n}(w) \cap 2\mathbb T^* 0| = |F_{n}(w) \cap 0\mathbb T^* 1|$$ and $$|F_{n}(w) \cap 2\mathbb T^* 2| = |F_{n}(w) \cap 1\mathbb T^* 1|.$$
    \end{proof}

    
    \begin{lemma}\label{lm:borders}
        Let $w$ be a ternary minimal complexity word of type III. Then for each $n \geqslant 1$ we have either $$|F_{n}(w) \cap 2\mathbb T^*| = |F_{n}(w) \cap 1\mathbb T^*|$$ or $$|F_{n}(w) \cap 2\mathbb T^*| = |F_{n}(w) \cap 1\mathbb T^*| + 1.$$
    \end{lemma}
    \begin{proof}
        The statement is straightforward for $n = 1$; consider $n \geqslant 2$. Since $|F_{n}(w) \cap 1\mathbb T^*| = |F_{n}(w) \cap\mathbb T^*2| = |F_{n - 1}(w) \cap\mathbb T^*1| = |F_{n-1}(w) \cap 2\mathbb T^*|$, the claim follows from the fact that we have exactly one right special factor of length $n - 1$.
        %
%
%
    \end{proof}

    \begin{proposition}\label{pr:abelian-four-ends-zero}
        Let $w$ be a ternary minimal complexity word $w$ of type III and $n$ be an integer. Then we have $p^{ab}_w(n) = 4$ if and only if the last column of the lexicographic array for length $n$ is of the form $0^+1^+2^+0^+$.
    \end{proposition}
    \begin{proof}
        First, we prove the sufficiency of the condition. The form of the last column of the lexicographic array implies that rule 3 is applied after rule 2. Thus, we have three cases depending on the order of the rules:

        Case 1. Rules are applied in the following order: 2 -- 3 -- 4. Lemma \ref{lm:abelian-rules} says that the abelian complexity in this case is 4.

        Case 2.  Rules are applied in the following order:  4 -- 2 -- 3. In this case factors that start with $0$ or $1$ could end only with $0$, while there exist factors from $F_{n}(w) \cap 2\mathbb T^*$ that end with $0$, $1$ and $2$. So, we have a factor from $F_{n}(w) \cap 2\mathbb T^* 0$, and using $\eta$ from Lemma \ref{lm:eta} we obtain a factor from $F_{n}(w) \cap 0\mathbb T^*1$, which is not possible since in this case all factors beginning with 0 also end with 0.

        Case 3.  Rules are applied in the following order: 2 -- 4 -- 3. Due to the form of the last column, we have a factor from $F_{n}(w) \cap 2\mathbb T^* 0$. By Lemma \ref{lm:eta}, we also have a factor from $F_{n}(w) \cap 0\mathbb T^* 1$.
        Due to the form of the last column and the order of the rules, we have that every factor from $F_{n}(w) \cap 1\mathbb T^*$ ends with $1$ and that the set $F_{n}(w) \cap 2\mathbb T^* 1$ is non-empty. As the set $F_{n}(w) \cap 0\mathbb T^* 1$ is also non-empty, we have $|F_{n}(w) \cap 2\mathbb T^*| = |F_{n}(w) \cap \mathbb T^* 1| \geqslant |F_{n}(w) \cap 1\mathbb T^*| + 2$, which is impossible due to Lemma \ref{lm:borders}.

        We now prove the necessity of the condition. We know that the abelian complexity is $4$, and hence Lemma \ref{lm:abelian-rules} implies that at most two cases are possible for the order of rules:

        Case 1. Rules are applied in the order 2 -- 3 -- 4. If the last column is not of the form $0^+1^+2^+0^+$, then in this case it could only be of the form $2^+0^+1^+2^+$. Applying $\eta$ to the set of factors  $F_{n}(w) \cap 2\mathbb T^* 2$ (resp., $F_{n}(w) \cap 0\mathbb T^* 2$), we obtain the set of factors $F_{n}(w) \cap 1\mathbb T^* 1$ (resp., $F_{n}(w) \cap 1\mathbb T^* 0$). Hence every factor from $F_{n}(w) \cap \mathbb T^* 1$ starts with $1$ and the sets  $F_{n}(w) \cap 1\mathbb T^* 0$ and $F_{n}(w) \cap 1\mathbb T^*2$ are non-empty, i.e., $|F_{n}(w) \cap 1\mathbb T^*| \geqslant |F_{n}(w) \cap \mathbb T^* 1| + 2 = |F_{n}(w) \cap 2\mathbb T^*| + 2$, which contradicts Lemma \ref{lm:borders}. 
        
        Case 2. Rules are applied in the order 4 -- 3 -- 2. In this case the last column has to be of the form $1^+2^+0^+1^+$. Every factor from $F_{n}(w) \cap \mathbb T^* 2$ starts with $2$ and the sets $F_{n}(w) \cap 2\mathbb T^* 1$ and $F_{n}(w) \cap 2\mathbb T^* 0$ are non-empty, i.e., $|F_{n}(w) \cap 2\mathbb T^*| \geqslant |F_{n}(w) \cap \mathbb T^* 2| + 2 = |F_{n}(w) \cap 1\mathbb T^*| + 2$, which contradicts Lemma \ref{lm:borders}. So, this case is impossible.
    \end{proof}

    \begin{theorem}
		Let $w$ be a word obtained by applying a morphism  
		$\psi\colon \begin{cases}
			0 \mapsto 0\\
			1 \mapsto 12
		\end{cases}$ to a Sturmian word, such that 
        the lexicographically smallest  factor of length $n$ ends with 0. Then $w$ satisfies the universal group complexity property for the length $n$.   
		
    \end{theorem}

\begin{proof}
			Follows from Lemma \ref{lm:four}, Propositions \ref{pr:no-four} and \ref{pr:abelian-four-ends-zero}.
\end{proof}



Recall that we denote by $C_n$ the cyclic group generated by the cycle $(1\dots n)$, i.e.,   $C_n=\left<(1\dots n)\right>$.

The following lemma provides a useful fact about group complexity for Sturmian words relatively to cyclic groups $C_n$:
					 
\begin{lemma}[\cite{CASSAIGNE201736}]\label{lm:cyclic-sturmian}
For each Sturmian word $s$ and each length $n$ such that there exists a bispecial factor of $s$ of length $n-2$, one has $p_s^{C_n} = 2$. Moreover, one class of factors has cardinality $n$ and the other one is singleton.
\end{lemma}

We note here that every standard factor of a Sturmian word $s$ is of the form $pab$, where $p$ is a palidrome and $a, b$ are distinct letters (cf. Theorem 2.2.4 in \cite{Lothaire_2002}). Moreover,  $p$ as a bispecial factor of $s$ (follows from Corollary 2.2.28 in \cite{Lothaire_2002}), hence we have the following corollary from Lemma \ref{lm:cyclic-sturmian}:
\begin{corollary}\label{cr:cyclic-sturmian}
For each Sturmian word $s$ and each length $n$ such that there exists a standard factor of $s$ of length $n$, one has $p_s^{C_n} = 2$. Moreover, one class of factors has cardinality $n$ and the other one is singleton.
\end{corollary}

					
Note that for every Sturmian word $s$ there exists infinitely many integers $n$ that satisfy Corollary \ref{cr:cyclic-sturmian}. 
					
					Now we prove a similar lemma for minimal complexity words of type III; moreover, we need only odd lengths:
					
					\begin{lemma}\label{lm:cycle} For each word $w$ obtained by applying a morphism  
						$\psi\colon \begin{cases}
			0 \mapsto 0\\
			1 \mapsto 12
		\end{cases}$ to a Sturmian word, there exist infinitely many odd integers $n$ such that $p_w^{ab}(n) = 3$ and $p_w^{C_n} = 3$. Moreover, there are three abelian classes: two classes are singleton, and the third class contains all cyclic shifts of one word.\end{lemma}
					
					\begin{proof}
%
Let $s$ be a Sturmian word and $n$ be a length from {Corollary \ref{cr:cyclic-sturmian}}. Suppose in addition that $n$ is large enough 
so that the corresponding standard factor is $s_m$ for $m \geqslant 3$ in the standard sequence of $s$. 
We call such $n$ \emph{suitable}. 
%
%
As for each standard factor $u=s_m$ of a Sturmian word satisfying the above condition $m\geqslant 3$ its square $uu$ is a factor of this Sturmian word (see, e.g., Lemma 2.2.32 in \cite{Lothaire_2002}), we have that $\psi(uu)$ is a factor of $\psi(s)$. Inside this factor we can find all $n+k$ cyclic shifts of the factor $\psi(u)$, where $k$ is the number of $1$'s in $\psi(u)$. 

Now, as we have a minimal factor complexity word, its factor complexity for length $n + k$ is $n+k+2$. We just proved that we have $n+k$ cyclic shifts of one word as factors; they are in the same group complexity class relatively to the group $C_n$. The remaining two factors must form singleton classes, as the group complexity is at least 3 by Lemma \ref{lm:abelian}.

						Now we prove that for infinitely many suitable lengths $n$ the length $n + k$ is odd. Let $\alpha$ be the slope of the word $s$ and $\frac{p_m}{q_m}$ be its convergents. 
                        Then the numbers $p_m + q_m$ are suitable lengths for $m\geqslant 3$, i.e., we can take $n=p_m + q_m$ and $k = p_m$  (see Subsection \ref{subsec:cont_frac}).     
                        By \eqref{eq:convergents}, the number $q_m$ cannot be even for consecutive values of $m$, hence $n + k = 2p_m + q_m$ is odd infinitely many times.			
	\end{proof}

    Given a word $w$ obtained by applying a morphism  
		$\psi\colon \begin{cases}
			0 \mapsto 0\\
			1 \mapsto 12
		\end{cases}$ to a Sturmian word.

        Suppose that $n$ is the length for which the conditions of Lemma \ref{lm:cycle} are satisfied. Consider the \emph{reduced lexicographic array}, i.e. the lexicographic array without the two factors forming their own abelian classes.
        We then define a \emph{box} as a fragment of the reduced lexicographic array of the form \begin{tabular}{ccc}
			0 & 1 & 2 \\
			1 & 2 & 0 \\
		\end{tabular}. The box does not have to be connected if it start in the last or in the penultimate 
        column and ends in the first or in the second column; it could also start in the last row and end in the first row.

	\begin{lemma} Let $w$ be a ternary minimal complexity words of type III and let $n$ be an integer satisfying the conditions of Lemma \ref{lm:cycle}. Then for each pair of consecutive factors from the reduced 
    lexicographic array (including the last one and the first one) their distinct elements form a box. Moreover, all such boxes start in different columns. \end{lemma}
	
    \begin{proof} 
    Consider two consecutive factors of the reduced lexicographic array (recall that all factors in the reduced lexicographic array are abelian equivalent). First suppose that there are no factors forming a singleton abelian class between them in the non-reduced lexicographic array (see Lemma \ref{lm:cycle}). Since rules 2--4 change the abelian class, this means that we applied rule 1 once, as in the non-reduced lexicographic array.

     Now suppose that between two factors that are consecutive in the reduced lexicographic array, there are other factors in the non-reduced lexicographic array. By Lemma \ref{lm:cycle}, we can have one or two such factors, and their abelian classes are singleton. 
     
     We start with the case when we have one such factor. If we applied rule $12\alpha \mapsto 20 \alpha$, then we lost one occurrence of 1, and we have to gain it back applying the next rule. This can be done only using rule $\alpha0 \mapsto \alpha1$. Note that applying other rules is impossible since it does not lead to the previous abelian class. 
     If we applied rule $\alpha0 \mapsto \alpha1$, the argument is symmetric: we lost one occurrence of 1, so we need to apply rule $12\alpha \mapsto 20 \alpha$ to gain it back. In either case  our two consecutive factors differ by a box starting in the last symbol.  If we applied rule $\alpha01 \mapsto \alpha12$, then with the next step we should lose an occurrence of 2, but we do not have such a rule. 

     We now show that the case of two factors with singleton abelian classes between two consecutive factors in the reduced lexicographic array is  impossible. If we have two consecutive singleton factors, it means that we applied rules 2,3 and 4 in some order, each of them exactly once. However, applying these three rules in no matter which order does not give the initial abelian class.

If we have two boxes starting at the same position, then the corresponding column is not of the form described in Corollary  \ref{cor:each_col}.

Now we study what happens between the first and the last row of the lexicographic array. Consider the left special factor $v$ of length $n-1$; it is extended to the left by 0 and 2, so $0v$ and $2v$ are factors in the lexicographic array; moreover, $0v$ is the first one and $2v$ is the last one (it follows from the properties of left special factors of Sturmian word: extended to the left by 0 and 1, they are lexicographically smallest and biggest factors of corresponding length). One of these words is singleton, and since the words have different numbers of 2's, this word is either the first one or the last one (as we proved, in the middle of the lexicographic array we can only have a singleton abelian class with the number of 1's different from factors in the reduced lexicographic array). 

We proved that we apply rules 2 and 4 consecutively in some order. By Theorem \ref{th:scheme} we must apply rule 3 once. Since it changes the abelian class and as we proved above cannot be followed or preceded by another rule changing the abelian class, we can only apply rule 3 in the first or in the last row of the lexicographic array. If we applied it to the first row, this means that the first row of the lexicographic array forms a singleton class. It is easy to see that the second row (which is the first row of the reduced lexicographic array, and which is obtained from $0v$ by applying rule 3) is obtained from the last row (which is of the form $2v$) by applying a box in the penultimate place. The situation is symmetric if rule 3 is applied to the first row.
		
	So, in the reduced lexicographic array, for each column we applied boxes three times; for the columns $1$, $n - 1$ and $n$ one of the boxes has been applied in the penultimate place, converting the last factor into the first one. 
\end{proof}

	\begin{theorem}\label{th:transitive} Let $x$ be a ternary minimal complexity word of type III, $n\geqslant5$ be a length satisfying the conditions of Lemma \ref{lm:cycle}, and $G\leqslant S_n$. If $p^G_x=4$, then $G$ acts transitively on the set $\{1,\ldots, n\}$. 
    \end{theorem}

    \begin{proof}
		Assume that $G$ does not act transitively; we will prove that then the abelian complexity is at least $5$.
		
		If $G$ has more than two orbits, we let $k$, $0<k<n$, denote the size of one of the orbits. We now construct a new group $G'$ by adding to $G$ permutations, so that $G'$ contains exactly two orbits, and moreover $G' = S_k \times S_{n - k}$. We have $p^{G'}_x(n)\leqslant p^{G}_x(n)$ by the definition of group complexity.
				
		Since $n$ is odd by the conditions of Lemma \ref{lm:cycle}, there exist two consecutive positions (we consider the last and the first one to be consecutive) such that they are in the same orbit. Since all factors in the reduced lexicographic array are cyclic permutations of each other, we can assume that these two positions are $1$ and $2$. We can also assume that the third position is from the other orbit; otherwise we can apply a permutation from $C_n$, or, equivalently, consider a conjugate subgroup.
		
		Case $1^{\circ}$: Position 4 is from the first orbit, i.e., we have orbits $1121\ldots$. 
		
		Consider the two boxes starting in the first and in the second positions. We need to understand when they change the Parikh vector of the first orbit. The first box does it as follows: $(a, b, c) \mapsto (a - 1, b, c + 1)$, and second one $(d, e, f) \mapsto (d, e + 1, f - 1)$. We claim that we have at least three distinct Parikh vectors among these four triples, and thus the group complexity is at least 5, taking into account the two factors forming singular abelian classes. Inside each pair the triples are distinct, so, if we have only two distinct Parikh vectors, then either $(a, b, c) = (d, e, f)$, or $(a, b, c) = (d, e + 1, f - 1)$. In the first case $(a, b, c), (a - 1, b, c + 1), (a, b + 1,  c - 1)$ are distinct, and in the second case $(a, b, c), (a - 1, b, c + 1), (a, b - 1, c + 1)$ are distinct.
		
		Case $2^{\circ}$. Position 4 is from the second orbit, i.e., we have orbits $1122\ldots$.
		
		The same two boxes as in the previous case give us pairs $(a, b, c) \mapsto (a - 1, b, c + 1)$ and $(d, e, f) \mapsto (d - 1, e + 1, f)$. If $(a, b, c) = (d, e, f)$, then the triples $(a, b, c), (a - 1, b, c + 1), (a - 1, b + 1, c)$ are distinct. If $(a, b, c) = (d - 1, e + 1, f)$, then the triples $(a, b, c), (a - 1, b, c + 1), (a + 1, b - 1, c)$ are distinct.
	\end{proof}
	
	\section{Periodic words}\label{subsec:purper}

    In this section, we study universal group complexity property for eventually periodic words. For an eventually periodic word we provide a criterion for determining whether the word has the universal group complexity property. This criterion can be checked algorithmically.

For integers $n$ and $m$ with $m<n$ and a subgroup $G$ of the symmetric group $S_n$, we let  $G \times S_m$  denote the subgroup of $S_{n+m}$ that acts on the first $n$ symbols as $G$, and permutes the last $m$ symbols in any way independently of the first symbols. 


    It is easy to see that the factor complexity function of an eventually periodic word $u$ increases until reaches its maximum at some point $n_0$, and it is constant for $n\geqslant n_0$. We denote by $n_{\mathrm{max}}(u)$ the smallest integer $n_0$ satisfying the above description: $n_{\mathrm{max}}(u) = \min\argmax\limits_n\{p_u(n)\}$. Note that if an eventually periodic word $u$ has a preperiod $\alpha$ and a period $\pi$, then $n_{\mathrm{max}}(u) \leqslant \alpha + \pi$.
    
    \begin{lemma}\label{lm:per} 
    Let $u$ be an eventually periodic word with period $\pi$, and $n$ be an integer such that  $n \geqslant n_{\mathrm{max}}(u)$. Then for each subgroup $G$ of $S_n$ we have $p^G_u = p^{G \times S_\pi}_u$.
    \end{lemma}
	
    \begin{proof} By the definition of $n_{\mathrm{max}}(u)$, we have $p_u(n) = p_u(n + \pi)$. The lexicographic array for the length $n+\pi$ is obtained from the lexicographic array for the length $n$ by extending it to the right by a $p_u(n) \times \pi$-rectangle filled with cyclic shifts of the same word. We have $p_u^{G \times S_n} \geqslant p_u^G$. Indeed, since on the first $n$ symbols the group $G \times S_n$ behaves as $G$, we have that if two words of length $n$ are in distinct $G$-classes, then their extension to the right of length $n+\pi$  are also in distinct $(G\times S_{\pi})$-classes. On the other hand, if two words of length $n$ are in the same $G$-class, then their extensions  to the right of length $n+\pi$ are in the same $(G\times S_{\pi})$-class, since $S_{\pi}$ permutes in any way the last $\pi$ symbols.
    \end{proof}
	
	\begin{corollary}\label{cr:finite} Let $u$ be an eventually periodic word with period $\pi$. If the universal group complexity property holds for lengths up to $n_{\mathrm{max}}(u) + \pi - 1$, then it holds for all lengths.\end{corollary}

    \begin{proof}
	By Lemma \ref{lm:per}, the set of values of group complexity for lengths $k \geqslant n_{\mathrm{max}}(u) + \pi$ is the same as for  lengths $k - \pi \geqslant n_{\mathrm{max}}(u)$. 
    Factor complexities for such lengths are equal, and the partition of factors into abelian classes does not change after concatenating factors with cyclic shifts of the same word. 
    \end{proof}

	\begin{example} Here we provide an example of a periodic word with universal group complexity. Consider the word $w=(001232)^\omega$. Its period $\pi$ is equal to $6$, and since $p_w(2) = 6$, we have $n_{\mathrm{max}}(w) = 2$. By Corollary \ref{cr:finite}, to show the universal group complexity property, it is enough to check it for lengths up to $7$:

    \begin{itemize} 
    \item For lengths $1$--$4$, one can directly  verify that the factor and the abelian complexities differ by at most 1. So, the universal group complexity property holds for these lengths.
    




    \item {For length $5$, we have $p^{ab}_w(5) = 4, p^{C_5}_w(5) = 5$, and $p_w(5) = 6$.}

    \item{For length $6$, we have $p^{ab}_w(6) = 1, p^{\left<(135)(246)\right>}_w(6) = 2, p^{\left<(14)(25)(36)\right>}_w(6) = 3, p^{S_{[2, 6]}}_w(6) = 4, p^{S_{[1, 2]} \times S_{[3, 6]}}_w(6) = 5$, and $p_w(6) = 6$.}

    \item {For length $7$, we have $p^{ab}_w(7) = 4, p^{C_7}_w(7) = 5$, and $p_w(7) = 6$.} \end{itemize}



    \end{example}

	\begin{remark} In fact, Corollary \ref{cr:finite} gives an algorithmic criterion for checking whether an eventually periodic word satisfies universal group complexity property.\end{remark}

	\section{Conclusions and open questions}

    In this paper, we consider the universal group complexity property for infinite words, which means that all values between the abelian and the factor complexities are achieved by the group complexity via some subgroups of the symmetric group. 
    We showed that Sturmian words satisfy this property; however, they are not the only ones. Besides that, we studied values of group complexity for ternary words of minimal complexity, and classified eventually periodic words with universal group complexity property. We propose several open problems:

    \medskip

    \textbf{Open problem 1.} Find a classification of words satisfying the universal group complexity property.
    
       \medskip

       It would be interesting to consider eventual universal group complexity property, i.e., when the property is satisfied starting from some length. The class of words satisfying the eventual universal group complexity is bigger than the class for all lengths (see Proposition \ref{pr:case2} and a remark after that).

    \medskip

    \textbf{Open problem 2.} A version of Open problem 1 for the eventual universal group complexity property.

    \medskip

A solution of the following problem would complete the classification from Theorem \ref{th:ternary} (see also Subsection \ref{subsec:four} for discussion).
    
    \medskip
    
    \textbf{Open problem 3.} For which ternary minimal complexity words of type III and for which length with abelian complexity 3 the group complexity can be equal to $4$?  

\medskip 

We suppose also that the classification from Theorem \ref{th:ternary} should be extendable to larger alphabet, although the proof is expected to be more technical.

\section*{Acknowledgements}

This work was supported by the Russian Science Foundation, project 25-21-00535. 
	
\bibliographystyle{plain}
\bibliography{references}
	
	
\end{document}